\documentclass[10pt, conference, compsocconf]{IEEEtran}
\usepackage{amsmath}
\usepackage{amsthm}
\usepackage{graphicx}
\usepackage{subfigure}
\usepackage{mathrsfs}
\usepackage{algorithm}
\usepackage{algorithmic}
\usepackage{color}
\usepackage{balance}
\usepackage{enumerate}
\usepackage{booktabs}
\usepackage{amsmath}
\usepackage{amssymb}
\usepackage[english, greek]{babel}

\newtheorem{theorem}{Theorem}
\newtheorem{lemma}{Lemma}
\newtheorem{corollary}{Corollary}

\newtheorem{definition}{Definition}

\newtheorem{example}{Example}

\begin{document} 
%
\IEEEoverridecommandlockouts
\selectlanguage{english}

\title{
Energy-Efficient Flow Scheduling and Routing with Hard Deadlines in Data Center Networks
}



%

\author{
\IEEEauthorblockN{
Lin Wang\IEEEauthorrefmark{1}\IEEEauthorrefmark{2}\IEEEauthorrefmark{6},
Fa Zhang\IEEEauthorrefmark{1}\IEEEauthorrefmark{2},
Kai Zheng\IEEEauthorrefmark{3},
Athanasios V. Vasilakos\IEEEauthorrefmark{4},
Shaolei Ren\IEEEauthorrefmark{5},
Zhiyong Liu\IEEEauthorrefmark{2}\IEEEauthorrefmark{7}
}
\IEEEauthorblockA{
\IEEEauthorrefmark{1}Key Laboratory of Intelligent Information Processing, Chinese Academy of Sciences} 
\IEEEauthorblockA{
\IEEEauthorrefmark{2}Institute of Computing Technology, Chinese Academy of Sciences, China} 
\IEEEauthorblockA{
\IEEEauthorrefmark{3}IBM China Research Lab, China} 
\IEEEauthorblockA{
\IEEEauthorrefmark{4}University of Western Macedonia, Greece} 
\IEEEauthorblockA{
\IEEEauthorrefmark{5}Florida International University, USA} 
\IEEEauthorblockA{
\IEEEauthorrefmark{6}University of Chinese Academy of Sciences, China}
\IEEEauthorblockA{
\IEEEauthorrefmark{7}State Key Laboratory for Computer Architecture, ICT, CAS, China}
}

\maketitle

\begin{abstract}
The power consumption of enormous network devices in data centers has emerged as a big concern to data center operators. Despite many traffic-engineering-based solutions, very little attention has been paid on performance-guaranteed energy saving schemes. In this paper, we propose a novel energy-saving model for data center networks by scheduling and routing ``deadline-constrained flows'' where the transmission of every flow has to be accomplished before a rigorous deadline, being the most critical requirement in production data center networks. Based on speed scaling and power-down energy saving strategies for network devices, we aim to explore the most energy efficient way of scheduling and routing flows on the network, as well as determining the transmission speed for every flow. We consider two general versions of the problem. For the version of only flow scheduling where routes of flows are pre-given, we show that it can be solved polynomially and we develop an optimal combinatorial algorithm for it. For the version of joint flow scheduling and routing, we prove that it is strongly NP-hard and cannot have a Fully Polynomial-Time Approximation Scheme (FPTAS) unless P=NP. Based on a relaxation and randomized rounding technique, we provide an efficient approximation algorithm which can guarantee a provable performance ratio with respect to a polynomial of the total number of flows. 

\end{abstract}

\section{Introduction}
\label{sec:intro}

Cloud computing has become a fundamental service model for the industry. In order to provide sufficient computing resources in clouds, large-scale data centers have been extensively deployed by many companies such as Google, Amazon and Microsoft. While providing powerful computing ability, those data centers are bringing a significant level of energy waste due to the inefficient use of hardware resources, resulting in both high expenditure and environmental concern.

Obviously, the servers should be the first target for energy reduction as they are the most energy-consuming component in a data center.   
By involving techniques such as Dynamic Voltage Frequency Scaling (DVFS) or hardware virtualization, the energy efficiency of servers has been improved to a large extent. As a result, the network device, as the second-place energy consumer, has taken a large portion in the total energy expenditure of a data center, bringing about urgent economic concerns over data center operators.

The problem of improving the network energy efficiency in data centers has been extensively explored (e.g., \cite{Heller_Seetharaman, Shang_Li, Wang_Yao,  Vasic_Bhurat-2011, Wang_Zhang-JSAC-2013}). Despite some energy-efficient network topologies for data centers, most of the solutions are concentrated on traffic engineering which aims to consolidate network flows and turn off unused network devices. The essential principle underlying this approach is that data center networks are usually designed with a high level of connectivity redundancy to handle traffic peak and that the traffic load in a data center network varies significantly over time. Due to the fact that the idle power consumed by the chassis usually takes more than half of a switch's total power consumption \cite{Mahadevan_Sharma}, turning off the switch during idle period should give the most power reduction in theory.

However, the practicality of these aforementioned solutions are quite limited because of the following two aspects: $\mathit{i})$ Most of the traffic-engineering-based approaches are ineluctably dependent on traffic prediction which seems not feasible or not precise enough \cite{Benson_Anand}. This is because the traffic pattern in a data center network largely depends on the applications running in the data center. Without precise traffic prediction, the network configuration generated by the energy-saving unit has to be updated frequently. Consequently, the network will be suffering from oscillation; $\mathit{ii})$ Saving energy leads to performance degradation. Most of the solutions only focus on energy efficiency without considering the network performance (e.g. throughput, delay). This will dramatically bring down the reliability of the network, which is not acceptable in practice as providing high performance is the primary goal in a network.

In order to overcome the above two limitations, we propose to view the network traffic from the application-level aspect instead of making use of the static network status (loads on links) that is rapidly monitored from the network or predicted (e.g. \cite{Heller_Seetharaman}). We observe that while the aggregate traffic load varies over time, the most critical factor that conditions the performance of many data exchanges in data centers is meeting flow deadlines  (\cite{Wilson_Ballani, Vamanan_Hasan, Hong_Caesar, Zats_Das, Alizadeh_Yang}). This is due to the fact that representative data center applications such as search and social networking usually generate a large number of small requests and responses across the data center that are combined together to perform a user-requested computation. As user-percieved performance is evaluated by the speed at which the responses to all the requests are collected and delivered to users, short or guaranteed latency for each of the short request/response flow is strongly required. Given a threshold for tolerable response latency, the system efficiency will be definitely conditioned by the number of flows whose deadlines are met (that are completed within the time threshold).

Inspired by this observation, we consider to represent the networking requirements of applications as a set of deadline-constrained flows\footnote{If not specified, ``flow'' in this paper refers to a certain amount of data that has to be transmitted from a source to a destination on the network.} and we aim to design particular energy-efficient scheduling and routing schemes for them. Although the job scheduling on single or parallel processors with deadline constrains has been extensively studied, little attention has been paid on the job scheduling problem on a multi-hop network \cite{Mao_Koksal}, especially with the objective of optimizing the energy consumption. To the best of our knowledge, this is the first solution that theoretically explores energy-efficient schemes by scheduling and routing deadline-constrained flows in data center networks.

To summarize our main contributions in this paper:

$\mathit{i})$ We describe the deadline-constrained network energy saving problem and provide comprehensive models for two general versions of this problem -- Deadline-Constrained Flow Scheduling (DCFS) and Deadline-Constrained Flow Scheduling and Routing (DCFSR); $\mathit{ii})$ We show that DCFS can be optimally solved in polynomial time and we propose an optimal combinatorial algorithm for it; $\mathit{iii})$ We show by in-depth analysis that solving DCFSR is strongly NP-hard and cannot have a fully Polynomial-Time Approximation Scheme (FPTAS) unless P=NP; $\mathit{iv})$ We provide an efficient approximation algorithm which solves the problem with a provable performance ratio.

The remainder of this paper is organized as follows. Section~\ref{sec:model} presents the modeling for the deadline-constrained network energy saving problem where two versions of the problem are introduced. Section~\ref{sec:dts} discusses the DCFS problem where an optimal combinatorial algorithm is provided to solve it. Section~\ref{sec:dtsr} discusses the DCFSR problem and presents some complexity and hardness analysis. Section~\ref{sec:approx} presents an approximation algorithm with guaranteed performance ratio for DCFSR where some numerical results are also provided. Section~\ref{sec:related} summarizes related work and section~\ref{sec:conc} concludes the paper.

\section{The Model}
\label{sec:model}

Based on some preliminary definition, we provide the general modeling for the deadline-constrained network energy saving problem in this section. 

\subsection{The Data Center}

We model a data center as a distributed computing system where a set of servers is connected with a network $\mathcal{G}=(\mathcal{V},\mathcal{E})$ where $\mathcal{V}$ is the set of nodes (switches and hosts) and $\mathcal{E}$ is the set of network links. We assume all the switches, as well as all the links, in $\mathcal{V}$ are identical which is reasonable because advanced data center networks such as fat-tree \cite{Al-Fares_Loukissas-FatTree-2008} or BCube \cite{Guo_Lu-BCube-2009} are usually conducted on identical commodity switches. We use the classical queueing model for links, that is, a link is modeled as a forwarding unit with buffers at its two ends. When the switch finishes processing a data packet, the egress port for this packet will be determined and this packet will be injected into the buffer of the egress link. The packets that queue in the buffer will be transmitted in order according to some preset packet scheduling policy.

We consider the power consumption of network components such as ports and links which are the main power consumers that can be manipulated for energy conservation.\footnote{The biggest power consumer in a switch, the chassis, cannot achieve power proportionality easily due to drastic performance degradation. Nevertheless, our approach is a complement and can be incorporated with switch-level power-down based solutions.} With a slight abuse of notation, the power consumption of the ports at the ends of a link is also abstracted into the power consumption of the link for the ease of exposition. For the power consumption model, we adopt the power function from \cite{Andrews_Fernandez-SS-2010} which is an integration of the power-down and the speed scaling model that has been widely used in the literature. For each link $e \in \mathcal{E}$, a power consumption function $f_e(x_e)$ is given to characterize the manner in which energy is being consumed with respect to the transmission rate $x_e$ of link $e$. We assume uniform power functions as the network is composed of identical switches and links. Formally, for every link we are given a function $f(\cdot)$ which is expressed by
\begin{equation}
\label{eqn:power_func}
f(x_e) = \left\lbrace
\begin{aligned}
&0 & x_e = 0& \\
&\sigma + \mu x_e^{\alpha} & 0 < x_e \leq C & 
\end{aligned}
\right., 
\end{equation}
where $\sigma, \mu$ and $\alpha$ are constants associated with the link type. Constant $\sigma$ is known as the idle power for maintaining link states, while $C$ is the maximum transmission rate of a link.
Normally, the power function $f(\cdot)$ is superadditive, i.e., $\alpha > 1$. In order to get rid of the network stability problem introduced by frequently togging on and off links, we assume that a link can be turned off only when it carries no traffic during the whole given period of time. Making this assumption also helps reduce the considerable power incurred by changing the state of a link, as well as the wear-and-tear cost.

\subsection{Applications}

We model an application as a set of deadline-constrained flows each of which consists of a certain amount of data that has to be routed from a source to a destination on the network within a given period of time. In order to avoid packet reordering at the destination end, we assume that each flow can only follow a single path. Nevertheless, multi-path routing protocols can be incorporated in our model by splitting a big flow into many small flows with the same release time and deadline at the source end and each of the small flows will follow a single path.

Let $[T_0,T_1]$ be a fixed time interval, during which a set $\mathcal{J} = \{j_1,j_2,\ldots,j_n\}$ of flows has to be routed on the network. Associated with each flow $j_i \in \mathcal{J}$ are the following parameters:
\begin{itemize}
\item $w_i$ amount of data that needs to be routed,
\item $r_i$, $d_i$ its release time and deadline respectively, and
\item $p_i$, $q_i$ its source and destination respectively.
\end{itemize}

Without loss of generality, we assume $T_0 = \min_{j_i\in \mathcal{J}}r_i$ and $T_1 = \min_{j_i\in \mathcal{J}}d_i$. In our setting, we allow preemption, i.e., each flow can be suspended at any time and recovered later. We define $S_i = [r_i, d_i]$ as the \emph{span} of flow $j_i$ and we say that $j_i$ is active at time $t$ if $t \in S_i$. The \emph{density} of flow $j_i$ is defined as $D_i = w_i/(d_i - r_i)$. A schedule is a set 
\begin{equation}
	\mathcal{S} = \left\{(s_i(t),\mathcal{P}_i)~|~\forall j_i \in \mathcal{J}, \forall t \in [r_i,d_i]\right\} 
\end{equation}
where $s_i(t)$ is the transmission rate chosen for flow $j_i$ at time $t$ and $\mathcal{P}_i$ is the set of links that are on the chosen path for carrying the traffic from this flow. A schedule is called feasible if every flow can be accomplished within its deadline following this schedule, i.e., $\mathcal{S}$ satisfies
\begin{equation}
\int_{r_i}^{d_i} s_i(t) dt = w_i, \forall j_i \in \mathcal{J}. 
\end{equation}
We define $\mathcal{E}_a$ as the set of active links where
\begin{equation}
\mathcal{E}_a = \{e \in \mathcal{E}~|~\exists t \in [T_0,T_1], x_e(t) > 0\}.
\end{equation}
Consequently, the total energy consumed by all the links during $[T_0,T_1]$ in a schedule $\mathcal{S}$ can be expressed by
\begin{equation}
\Phi_f(\mathcal{S}) = (T_1-T_0)\cdot|\mathcal{E}_a|\cdot\sigma + \int_{T_0}^{T_1} \sum_{e \in \mathcal{E}_a}\mu\left(x_e(t)\right)^{\alpha}dt  
\end{equation}
where $x_e(t)$ is the transmission rate of link $e$ at time $t$ and $x_e(t) = s_i(t)$ if flow $j_i$ is being transmitted on link $e$ at time $t$. Our objective is to find a feasible schedule that minimizes $\Phi_f(\mathcal{S})$. Depending on whether the routing protocol is given or not, we have two versions of this problem which we call DCFS (Deadline-Constrained Flow Scheduling) and DCFSR (Deadline-Constrained Flow Scheduling and Routing). We will discuss them separately in the following sections. 

\section{Deadline-Constrained Flow Scheduling}
\label{sec:dts}

In this section, we discuss the DCFS problem. Specifically, we model this problem as a convex program and show that it can be optimally solved. We then provide an optimal combinatorial algorithm for it.

\subsection{Preliminaries}

In DCFS, the routing paths for all the flows are provided. Routing the flows with these paths, each link will be assigned with a set of flows $\mathcal{J}_e = \{j_i~|~e \in \mathcal{P}_i\}$.
We omit those inactive links that satisfy $\mathcal{J}_e = \emptyset$ since they will never be used for transmitting data. Thus the set of active links is $\mathcal{E}_a = \mathcal{E} \setminus \{e \in \mathcal{E}~|~\mathcal{J}_e = \emptyset\}$. As all the links in $\mathcal{E}_a$ will be used, we simplify the problem by replacing the power consumption function with $g(x_e) = \mu x_e^{\alpha}$. Consequently, the objective of the problem becomes to find a feasible schedule $\mathcal{S}$ such that
\begin{equation}
\Phi_g (\mathcal{S}) = \int_{T_0}^{T_1} \sum_{e \in \mathcal{E}_a}g\left(x_e(t)\right)dt 
\end{equation}
is minimized. For the sake of tractability, we first consider the case where the routing path for each flow is a virtual circuit. That is, when a flow is being routed, all the links on the routing path of this flow will be totally occupied by the packets from this flow. Nevertheless, we will show that this assumption is generally true in the optimal solution and it can be realized by assigning priorities to the packets from each flow in a packet-switching network. 

We define the \emph{minimum-energy schedule} as the schedule that minimizes the total power consumption but may not satisfy the maximum transmission rate constraint on each link. Then, we introduce the following lemmas.
\begin{lemma}
\label{lm:unique}
The minimum-energy schedule will use a single transmission rate for every flow.
\end{lemma}

\begin{proof}
We prove it by contradiction. Suppose we are given an instance of DCFS and a minimum-energy schedule $\mathcal{S}$ for this instance where we have two different transmission rates for only one of the flows. For the ease of exposition, we assume that the time interval\footnote{It can also be extended to the case where we have more than one interval for routing this flow as we only focus on this flow.} that this flow is being routed on the network is $[a,b]$ in $\mathcal{S}$ and there is a time point $t$ ($a < t < b$) such that the transmission rate is $s_1$ in interval $[a,t]$ and $s_2$ in interval $[t,b]$, respectively. Now, instead of using $s_1$ and $s_2$ in the two different intervals, we propose another schedule $\mathcal{S}'$ where we use a single rate $\frac{(t-a)s_1 + (b-t)s_2}{b-a}$ through the whole interval $[r,d]$. Due to the convex property of $g(\cdot)$, it is easy to verify that
\begin{equation}
\mu (t-a)s_1^{\alpha} + \mu(b-t)s_2^{\alpha} > \mu\left(\frac{(t-a)s_1 + (b-t)s_2}{b-a}\right)^{\alpha}. \nonumber
\end{equation}
The above inequality is equivalent to $\Phi_g(\mathcal{S}) > \Phi_g(\mathcal{S}')$, which contradicts the assumption that $\mathcal{S}$ is optimal. This completes the proof.
\end{proof}

\begin{lemma}
\label{lm:minimum}
The minimum-energy schedule will choose an as small as possible transmission rate for each flow such that the deadlines of flows can be guaranteed.
\end{lemma}
\begin{proof}
We focus on one flow with an amount $w$ of data that has to be routed in time interval $[r,d]$. The routing path for this flow is denoted by $\mathcal{P}$ and the number of links in $\mathcal{P}$ is given by $|\mathcal{P}|$. Using Lemma~\ref{lm:unique}, we assume a single transmission rate $s$ is given to process this flow. The total energy consumed by the links for routing this flow can be expressed by $\Phi_g = \mu s^{\alpha} \cdot w/s = \mu \cdot w \cdot s^{\alpha - 1}$. As long as $\alpha > 1$, $\Phi_g$ is minimized when we have the minimum transmission rate $s$ for this flow such that the deadlines of all the flows can be satisfied. In this sense, the minimum-energy schedule will use the minimized transmission rate for each flow.
\end{proof}

Following the above lemma, we observe that as long as there are feasible schedules, the minimum-energy schedule is feasible. In other words, the minimum-energy schedule is also the optimal schedule for DCFS. Equivalently, the maximum transmission rate constraint $x_v(t) \leq C$ can be relaxed in DCFS. In the remainder of this section, we will omit that constraint.

\subsection{Problem Formulation}
We denote the transmission rate for flow $j_i$ as $s_i$ according to Lemma~\ref{lm:unique}.
The DCFS problem can be formulated as the following convex program.
\begin{equation}
\begin{aligned}
&(P_1)~~~~\min  \sum_{e \in \mathcal{E}_a} \sum_{j_i \in \mathcal{J}_v} w_i \cdot \mu s_i^{\alpha-1}  \\
\text{subject to}  \\
& \sum_{j_i \in \mathcal{J}'} \frac{w_i}{s_i} - (\max_{j_i \in \mathcal{J}'}d_i - \min_{j_i \in \mathcal{J}'}r_i) \leq 0 & \mathcal{J}' \subseteq \mathcal{J}_e  \\
& s_i > 0& \forall j_i \in \mathcal{J}
\end{aligned}
\nonumber
\end{equation}
The total transmission time and the total energy consumption of flow $j_i$ are $w_i/s_i$ and $w_i/s_i \cdot \mu s_i^{\alpha}=w_i\cdot\mu s_i^{\alpha-1}$, respectively. The first constraint forces that for an arbitrary link $e$, all the flows in any subset of $\mathcal{J}_e$ has to be processed before the last deadline of the flows in that subset. The second constraint represents that the transmission rate for each flows has to be larger than 0. It is easy to verify that program ($P_1$) is convex because the objective function is convex (as we assume $\alpha > 1$) while all the constraints are linear. As a result, the DCFS problem can be solved optimally in polynomial time by applying the Ellipsoid Algorithm \cite{Nesterov-Ellipsoid-1994}. However, as the Ellipsoid Algorithm is not practically used due to its high complexity in typical instances, we aim to construct an efficient combinatorial algorithm by exploring the characteristics of the minimum-energy schedule.

\subsection{An Optimal Combinatorial Algorithm}

We now provide a combinatorial algorithm which can always find the optimal schedule for DCFS. Before presenting the algorithm, we first give a characterization of the optimal schedule through the following example.
\begin{example}
Consider a line network whose topology is given in Fig.~\ref{fig:example_1}. The power consumption of the links is characterized by function $f(x_e) = x_e^2$. On this network we have two flows $j_1$ and $j_2$ that need to be routed. The details of the two flows are given by the following multi-tuples
\begin{eqnarray}
j_1 &\triangleq & (p_1=A,q_1=C,r_1=2,d_1=4,w_1=6), \nonumber\\
j_2 &\triangleq & (p_2=A,q_2=B,r_2=1,d_2=3,w_2=8). \nonumber
\end{eqnarray}
\end{example}
\begin{figure}[t!]
\centering
\includegraphics[scale=0.5]{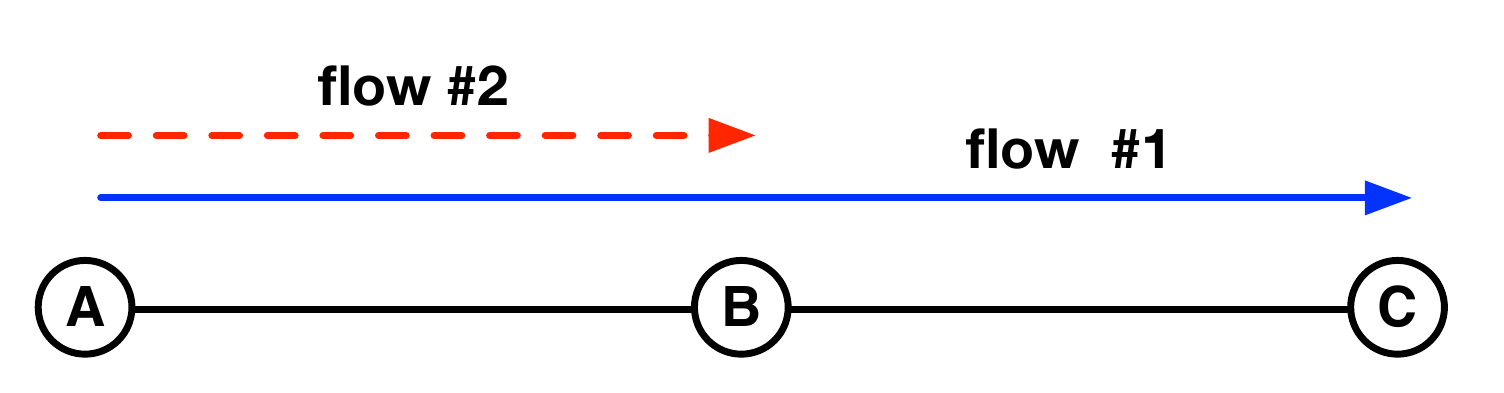}
\caption{\label{fig:example_1} A line network consisting of three nodes connected by two links. }
\end{figure}
According to Lemma~\ref{lm:unique}, we denote the transmission rates for $j_1$ and $j_2$ as $s_1$ and $s_2$, respectively. Consequently, we have the following three constraints $6/s_1 \leq d_1 - r_1 = 2$, $8/s_2 \leq d_2-r_2=2$ and $6/s_1 + 8/s_2 \leq d_1 - r_2=3$, while the objective function is $\Phi=2 \times 6\times s_1 + 8\times s_2$. It is easy to check that in the optimal schedule, $\sqrt{2} s_1 = s_2 = \frac{8+6\sqrt{2}}{3}$. We then construct an instance of the speed scaling problem on single processor (SS-SP) raised by Yao \emph{et al.} \cite{Yao_Demers_Shenker-YDS-1995}. Consider we have two jobs with required numbers of CPU cycles of $6\sqrt{2}$ and $8$ respectively, while the release times and deadlines are exactly the same as the two flows. Using the \textbf{Optimal-Schedule} algorithm (known as the \textbf{YDS} algorithm according to the authors' initials), the two jobs will be processed at the same speed of $\frac{8+6\sqrt{2}}{3}$ in time interval $[1,4]$. As a result, the objective value in the optimal schedule for this instance is exactly the same as the minimum $\Phi$ in our problem while the structure of the solution is also the same. Using this observation, we provide an optimal algorithm for solving the DCFS problem based on the \textbf{YDS} algorithm. 

We first construct from the DCFS problem a variant of the SS-SP problem by introducing a virtual weight $w_i' = w_i \cdot (|\mathcal{P}_i|)^{1/\alpha}$ for each flow $j_i \in J$. First of all, we present some definitions which are extended from \cite{Yao_Demers_Shenker-YDS-1995}.
\begin{definition}
The intensity of an interval $I = [a,b]$ on a link $e$ is defined by
\begin{equation}
\delta(I,e) = \frac{\sum_{[r_i,d_i]\subseteq [a,b] \wedge j_i \in \mathcal{J}_e}w_i'}{a\sim b} 
\end{equation}
where $a \sim b$ denotes the available time in interval $[a,b]$.
\end{definition}
Intuitively, the following inequality has to be satisfied,
\begin{equation}
\int_{a}^{b} x_e(t) dt/(a \sim b) \geq \delta(I,e), 
\end{equation}
which means that $\delta(I,e)$ is a lower bound on the average transmission rate of link $e$ in any feasible schedule over the interval $[a,b]$.
\begin{definition}
If an interval $I^*=[a,b]$ maximizes $\delta(I^*,e)$ for any $e \in \mathcal{E}_a$, we call $I^*$ a critical interval and $e$ is the corresponding critical link.
\end{definition}
Now we present the main algorithm that generates optimal schedules greedily by computing critical intervals iteratively. The pseudo-code of the algorithm \textbf{Most-Critical-First} is shown in Algorithm~\ref{alg:greedy}.
\begin{algorithm}[!t]
\caption{\label{alg:greedy} \textbf{Most-Critical-First}}
\textbf{Input:} data center network $\mathcal{G}=(\mathcal{V},\mathcal{E})$, set of flows $\mathcal{J}_e$ for $e \in \mathcal{E}$ and virtual weights $w_i'$ for each flow\\
\textbf{Output:} transmission rate $s_i$ and transmission time interval $[r_i',d_i']$ for each flow $j_i \in \mathcal{J}$

\begin{algorithmic}[1]
\WHILE{$\exists e \in \mathcal{E}, \mathcal{J}_e \neq \emptyset$}
	\STATE Find the critical interval $I^*$ and the critical link $e$. The flows in this interval can be represented by $\mathcal{J}^* = \{j_i~|~[r_i,d_i]\subseteq I^* \wedge e \in \mathcal{P}_i \}$ and without loss of generality, 
	\begin{equation}
		I^* = [a,b] =[\min_{j_i \in \mathcal{J}^*}r_i,\max_{j_i \in \mathcal{J}^*}d_i]. \nonumber
	\end{equation}
	\STATE Schedule the flows in $\mathcal{J}^*$ with the Earliest Deadline First (EDF) policy using transmission rate 
	\begin{equation}
		s_i = \frac{\sum_{j_i \in \mathcal{J}^*} w_i'}{(|\mathcal{P}_i|)^{1/\alpha}(a \sim b)} \nonumber
	\end{equation}
	for each flow $j_i \in \mathcal{J}^*$. The transmission time interval $[r_i',d_i']$ is also determined.
	\FOR{$j_i \in \mathcal{J}^*$}
		\STATE $\mathcal{J}_e \leftarrow \mathcal{J}_e \setminus j_i$ for $e \in \mathcal{P}_i$.
		\STATE For $e \in \mathcal{P}_i$, mark the time interval $[r_i',d_i']$ as unavailable on link $e$.
	\ENDFOR
\ENDWHILE
\end{algorithmic}
\end{algorithm}

The following theorem proves that a critical interval will determine a segment of the optimal schedule.
\begin{theorem}
\label{thm:interval-opt}
Let $I^*$ be a critical interval and $e$ be the corresponding critical link, algorithm \textbf{Most-Critical-First} can guarantee that the energy consumed by routing the flows in $\mathcal{J}^*$ is optimal.
\end{theorem}
\begin{proof}
Denoting the transmission rate for each flow $j_i \in \mathcal{J}^*$ as $s_i$, the total energy consumed by routing the flows in $\mathcal{J}^*$ is expressed by
\begin{equation}
\Phi_g(\mathcal{J}^*) = \sum_{j_i \in J^*} |\mathcal{P}_i| \times w_i \times s_i^{\alpha-1}. 
\end{equation}
According to the second constraint in program $(P_1)$ we have
\begin{equation}
\sum_{j_i \in \mathcal{J}^*} \frac{w_i}{s_i} -(a \sim b) \leq 0. 
\end{equation}
It is clear that in the optimal schedule, the above inequality is exactly an equality because of Lemma~\ref{lm:minimum}. Then, $\Phi_g(\mathcal{J}^*)$ can be minimized by using the method of Lagrange multipliers. By introducing a Lagrange multiple $\lambda$, we construct a function
\begin{equation}
\begin{aligned}
L(s_1,s_2,\ldots,s_{|J^*|},\lambda) = \Phi_g(\mathcal{J}^*) + \lambda (\sum_{j_i \in \mathcal{J}^*} \frac{w_i}{s_i} -(a \sim b)). 
\end{aligned}
\end{equation}
By setting $\bigtriangledown L(s_1,s_2,\ldots,s_{|J^*|},\lambda) = 0$, we have
\begin{equation}
\left( |\mathcal{P}_1|\right)^{1/\alpha}s_1 = \ldots = \left( |\mathcal{P}_{|\mathcal{J}^*|}|\right)^{1/\alpha}s_n.
\end{equation}
This is equivalent to solving an instance of the SS-SP problem as we explained in Example~\ref{fig:example_1}, where we treat each flow as a job with weight $w_i' = (|\mathcal{P}_i|)^{1/\alpha}w_i$. Using Theorem~$1$ provided in \cite{Yao_Demers_Shenker-YDS-1995}, we set the processing speed of all the jobs to the same value of $\sum_{j_i \in \mathcal{J}^*}w_i' / (a\sim b)$, which will give the optimal energy consumption for the SS-SP problem, as well as our problem. That is, $\Phi_g(\mathcal{J}^*)$ is minimized by setting
\begin{equation}
\left( |\mathcal{P}_1|\right)^{1/\alpha}s_1 = \ldots = \left( |\mathcal{P}_{|\mathcal{J}^*|}|\right)^{1/\alpha}s_n = \frac{\sum_{j_i \in \mathcal{J}^*}w_i'}{ (a\sim b)} 
\end{equation}
which is reflected in the algorithm.
\end{proof}
Actually, algorithm \textbf{Most-Critical-First} solves a variant of the SS-SP problem based on the YDS algorithm. Consequently, the following result follows quickly from Theorem~\ref{thm:interval-opt}.
\begin{corollary}
The schedule produced by algorithm \textbf{Most-Critical-First} is optimal to the DCFS problem.
\end{corollary}

The time complexity of algorithm \textbf{Most-Critical-First} is bounded by $O(n^2|\mathcal{V}|)$. Note that the optimality of this algorithm is maintained based on the assumption that data in flows is routed exclusively through virtual circuits. We now show how to extend it to a packet-switching network: we assign a unique priority for all the packets from each flow according to the flow's starting time $r_i'$. That is, a flow $j_i$ with a smaller $r_i'$ will have a higher priority. This priority information can be encapsulated into the header of each packet and links will schedule those packets according to their priorities. 

\section{Deadline-Constrained Flow Scheduling and Routing}
\label{sec:dtsr}

In this section, we discuss the DCFSR problem. We aim at exploring the most energy-efficient scheduling and routing scheme for a given collection of flows. This problem is much harder than DCFS as we have to decide also the routing path for each flow, as well as the transmission rate. 

\subsection{Problem Formulation}

We observe that once we have the routing paths for all flows determined, finding the transmission rate for each flow is then the DCFS problem which can be optimally solved by algorithm \textbf{Most-Critical-First}. Let $\mathcal{P}_i$ denote the routing path for flow $j_i$. Keeping the notation we used before, the DCFSR problem can be formalized by the following program.
\begin{equation}
\begin{aligned}
&(P_2)~~~~\min \Phi_f \\
\text{subject to} \\
&\int_{r_i}^{d_i} s_i(t) dt \geq w_i  & \forall j_i \in \mathcal{J}\\
& s_i(t) \leq x_e(t) & \forall e \in \mathcal{P}_i \\
& 0 < x_e(t) \leq C  & \forall e \in \mathcal{E} \\
& s_i(t): \text{flow conservation}
\end{aligned}
\nonumber
\end{equation}
The first constraint represents that each flow has to be finished before its deadline. The second constraint means that the transmission rate of the flow that is being processed on a link $e$ cannot exceed the operation rate of that link, while the third one represents that the operation rate of a link has to be larger than zero and no larger than the maximum operation rate $C$. The flow conservation in the last constraint forces that $\mathcal{P}_i$ is a path connecting source $p_i$ and destination $q_i$ of flow $j_i$.

\subsection{Complexity and Hardness Results}

First, we provide the following definition and lemma as preliminaries based on our power consumption model. 
\begin{definition}
The \textbf{power rate} of a link $e$ ($x_e > 0$) is defined as the power consumed by each unit of traffic, i.e., $f(x_e)/x_e$.
\end{definition}
It can be observed that as long as the power rate of every link is minimized, the total power consumption of the network will be optimal. To minimize the power rate of a link, we show the following lemma.
\begin{lemma}
\label{lm:opt_rate}
Ideally, the optimal operation rate $R_{opt}$ for a link is given by $R_{opt} = \left( \frac{\sigma}{\mu(\alpha-1)} \right)^{1/\alpha}$.
\end{lemma}
Note that this operation rate is optimal in theory but is not always achievable in practice, as it can happen that $R_{opt}> C$. In general, we can prove that the decision version of DCFSR is NP-complete by providing the following theorem.
\begin{theorem}
Given a certain amount of energy $\Phi_0$, finding a schedule $\mathcal{S}$ for DCFSR such that $\Phi(\mathcal{S}) \leq \Phi_0$ is NP-complete.
\end{theorem}

\begin{proof}
This can be proved by a simple reduction from the $3$-partition problem which is NP-complete \cite{Garey_Johnson-1990}. Suppose we are given an instance of the $3$-partition problem with a set $\mathcal{A}$ of $3m$ integers $a_1, a_2,\ldots, a_{3m}$ where $\sum_{i=1}^{3m} a_i = mB$ and $a_i \in [B/4, B/2]$. The problem is to decide whether $\mathcal{A}$ can be partitioned into $m$ disjoint subsets, i.e., $\mathcal{A} = \mathcal{A}_1 \cup \mathcal{A}_2 \cup \ldots \cup \mathcal{A}_m$ and $\mathcal{A}_i \cap \mathcal{A}_j = \emptyset$ for any $i, j \in m$, such that every subset $\mathcal{A}_i$ consists of $3$ integers and $\sum_{a \in \mathcal{A}_i} a = B$. Based on this $3$-partition instance, we construct an instance of DCFSR as follows: we are given a network where two nodes (denoted as \emph{src} and \emph{dst}) are connected in parallel by $k$ ($k >> m$) links. Assume we are given a set $\mathcal{J}$ of $3m$ flows each of which has an amount $a_i$ ($i \in [1,3m]$) of data needed to be transmitted from \emph{src} to \emph{dst} on the network. All the flows arrive at the same time and the data transmission has to be finished in one unit of time. We assume $B < C$ and $\sigma = \mu (\alpha-1) B^{\alpha}$, i.e., $R_{opt} = B$ and we set $\Phi_0 = m\cdot \alpha\mu B^{\alpha}$. We will show that there is a schedule $\mathcal{S}$ such that $\Phi(\mathcal{S}) \leq \Phi_0$ if and only if $\mathcal{A}$ can be partitioned in the way as in the optimal solution of the $3$-partition instance.

On the one hand, if there exists a partition for the $3$-partition instance, we have a solution $\mathcal{S}$ for the DCFSR instance where the flows are transmitted by $m$ links each with an operation rate $B$ according to the partition and the energy consumption in this solution is $\Phi(\mathcal{S}) = m\cdot \alpha\mu B^{\alpha}$. According to Lemma~\ref{lm:opt_rate}, this solution is optimal since the power rate for each link in this solution is optimal. Hence, it satisfies that $\Phi(\mathcal{S}) \leq \Phi_0$. On the other hand, if we obtain a solution $\mathcal{S}$ for the DCFSR instance such that $\Phi(\mathcal{S}) \leq \Phi_0$. It can then be derived that exactly $m$ links will be used and each link will use an operation rate $B$. Otherwise, the total energy consumption $\Phi(\mathcal{S})$ will be larger than $\Phi_0$ as the average power rate of the used links is larger than $f(B)/B$, so $\Phi(\mathcal{S}) > f(B)/B\cdot mB=m\cdot \alpha\mu B^{\alpha}$. Accordingly, we can construct a partition for the $3$-partition instance. In a nutshell, finding a partition for $3$-partition is equivalent to finding a solution $\mathcal{S}$ for DCFSR such that $\Phi(\mathcal{S}) \leq \Phi_0$. 

The above reduction is based on the assumption that $R_{opt} < C$, which is not necessarily true in reality. However, in the case $R_{opt} > C$, we just set $B=C$ and $\Phi_0 = m(\sigma + \mu C^{\alpha})$, and the same reduction can be built in a similar way.
\end{proof}
Then, it follows directly that
\begin{corollary}
Solving the DCFSR problem is strongly NP-hard.
\end{corollary}
As a result, the DCFSR problem can only be solved by approximating the optimum. When we say that an algorithm approximates DCFSR with performance ratio $\gamma$, it means that the energy consumption in the solution produced by this algorithm is at most $\gamma$ times the minimum energy consumption. Given these, we aim at designing an algorithm to approximate the optimum with ratio $\gamma$ as small as possible. Unfortunately, for the case $R_{opt} > C$, which is likely to be the real situation as justified in \cite{Wang_Zhang-JSAC-2013}, the following theorem shows that $\gamma$ cannot be as small as we want since there is a lower bound for it.

\begin{theorem}
There exists an instance of problem DCFSR such that no approximation algorithm can guarantee a performance ratio smaller than $\frac{3}{2} \left(1 + \frac{(2/3)^{\alpha}-1}{ \alpha}\right)$ unless P=NP.
\end{theorem}

\begin{proof}
We prove this theorem by showing a gap-preserving reduction from the partition problem which is NP-complete. Suppose we are given an instance of the partition problem with a set $\mathcal{A}$ of $n$ integers. Assuming $\sum_{a \in \mathcal{A}} a = B$, the problem is whether it is possible to find a subset $\mathcal{A}' \subset \mathcal{A}$ such that $\sum_{a \in \mathcal{A}'} a = B/2$.  We now construct an instance of the DCFSR problem as follows:  we consider also the same network as the one in the previous proof where we assume $m > 2$ and the capacity of each link is given by $C = B/2$. We are also given a set $\mathcal{J}$ of flows each of which requires to route an amount $w_i$ of data from \emph{src} to \emph{dst} and $w_i = \mathcal{A}[i]$ for $j_i \in \mathcal{J}$. These flows arrive at the system at the same moment and have to be accomplished in one unit time. We denote by $\Phi_{opt}$ the optimal solution of the DCFSR instance, which represents the total energy consumed by the active links for tranmisstting the flows. Then, the following properties are preserved.
	\begin{align}
		\exists \mathcal{A}' \subset \mathcal{A}, \sum_{a \in \mathcal{A}'} a = B/2 	&\Longrightarrow \Phi_{opt} = 2b + 2 \mu C^{\alpha}, \nonumber\\
		\nexists \mathcal{A}' \subset \mathcal{A}, \sum_{a \in \mathcal{A}'} a = B/2 	&\Longrightarrow \Phi_{opt} \geq 3b + 3\mu (2C/3)^{\alpha}.\nonumber
	\end{align}
	Comparing both optimal solutions, we obtain a ratio $\gamma$ where
	\begin{align}
		\gamma 	&= \frac{3\sigma + 3 \mu \left(2C/3\right)^{\alpha}}{2\sigma + 2 C^{\alpha}} \geq \frac{3 \mu C^{\alpha}(\alpha - 1) + 3 \mu\left(2C/3\right)^{\alpha}}{2 \mu C^{\alpha}(\alpha - 1) + 2 \mu C^{\alpha}} \nonumber\\
				&= \frac{3}{2} \left(1 + \frac{(2/3)^{\alpha}-1}{ \alpha}\right) 
	\end{align}
where the inequality is obtained by applying $\sigma \geq \mu C^{\alpha}(\alpha - 1)$.
Combining with the two properties we derived, it is easy to conclude that as long as there is an approximation algorithm solving DCFSR with a performance ratio better than $\gamma$, a subset $\mathcal{A}'$ in the partition problem can be found such that $\sum_{a \in \mathcal{A}'} a = B/2$. However, it is well known that the partition problem is NP-complete and cannot be solved by any polynomial time algorithm optimally. As a result, no algorithm can approximate DCFSR with a performance ratio better than $\gamma$ unless P=NP.
\end{proof}
This directly implies that the DCFSR cannot have Fully Polynomial-Time Approximation Schemes (FPTAS) under the conventional assumption that P$\neq$NP \cite{Vazirani-2004}. In the next section, we will provide an efficient approximation algorithm.

\section{An Approximation Algorithm for DCFSR}
\label{sec:approx}

We present an approximation algorithm for DCFSR in this section. This algorithm is based on a relaxation and randomized rounding based process. We show by both analysis and numerical results that this approximation algorithm can guarantee a provable performance ratio.

\subsection{The Algorithm}

We first provide the following preliminaries. We define $\mathcal{T} = \{t_0, t_1,\ldots,t_K\}$ to be the set of release times and deadlines of all the flows such that $t_{k_1} < t_{k_2}$ for any $0 \leq k_1 < k_2 \leq K$. It is clear that $t_0 = \min_{j_i \in \mathcal{J}} \{r_i\}$ and $t_K = \min_{j_i \in \mathcal{J}} \{d_i\}$. We denote by $I_k$ the time interval $[t_{k-1},t_k]$ for $1\leq k \leq K$, by $|I_k|$ the length of interval $I_k$ and by $\beta_k = |I_k|/(t_K-t_0)$ the fraction of time that interval $I_j$ takes from the whole time of interest. We also define $\lambda = (t_K-t_0)/\min_{k}|I_k|$.

We first relax the problem by making the following transformations such that the resulted problem can be optimally solved.
\begin{itemize}
\item The traffic load of flow $j_i$ is given by its density $D_i$. Flows can be routed simultaneously on any link;
\item Each flow can be routed through multiple paths;
\item The links in the network can be flexibly turned on and off at any moment.
\end{itemize}
We observe that the resulted problem can be decomposed into a set of subproblems in each interval $I_k$ as in each interval $I_k$, the traffic flows on the network are invariable. Actually, each such subproblem in an interval is a fractional multi-commodity flow (F-MCF) problem that is precisely defined as follows.
\begin{definition}[F-MCF]
For every active flow $j_i \in \mathcal{J}$ that satisfies $I_k \subseteq S_i$, a flow of traffic load $D_i$ has to be routed from $p_i$ to $q_i$. The objective is to route these flows on the network such that the total cost on the links is minimized, given cost function $f(\cdot)$ for every link.
\end{definition}

It is known that the F-MCF problem can be optimally solved by convex programming. Consequently, we obtain the fractional solution $y^*_{i,e}(k)$ which represents the proportion of the amount of the flow $j_i$ that goes through link $e$ in interval $I_k$. Absolutely, this solution is not feasible to the original DCFSR problem. Now we aim to transform this infeasible solution into a feasible one.

The transformation is accomplished by a randomized rounding process. Before that, we extract candidate routing paths for each flow following the Raghavan-Tompson \cite{Raghavan_Tompson} manner as follows. For each interval $I_k$ ($1\leq k \leq K$), we decompose the fractional solution $y^*_{i,e}(k)$ into weighted flow paths for each flow $j_i$ via the following procedure. We repeatedly extract paths connecting the source and destination of each flow $j_i$ from the subgraph defined by links $e$ for which $y^*_{i,e}(k) > 0$. For each extracted path $\mathcal{P}$, we assign a weight $w_\mathcal{P} = \min_{e \in \mathcal{P}} y^*_{i,e}(k)$ for $\mathcal{P}$ and the value of $y^*_{i,e}(k)$ on every link in $\mathcal{P}$ is reduced by $w_\mathcal{P}$. This path extracting process will terminate when every $y^*_{i,e}(k)$ becomes zero, which is guaranteed by the flow conservation constraint. As a result, we obtain a set $\mathcal{Q}_i(k)$ of paths for flow $j_i$ in interval $I_k$. We repeat this process for every interval and denote $\mathcal{Q}_i = \cup_{1\leq k\leq K} \mathcal{Q}_i(k)$ as the set of all the candidate paths for flow $j_i$ without duplication. Note that a path $\mathcal{P}$ may be used in more than one interval. We denote by $w_\mathcal{P}(k)$ the corresponding weights of $\mathcal{P}$ in different intervals. If $\mathcal{P}$ is not used in interval $I_k$, then $w_\mathcal{P}(k) = 0$.

Now we show how to choose a single path for each flow $j_i$ from the candidate paths $\mathcal{Q}_i$. For each path $\mathcal{P} \in \mathcal{Q}_i$, we assign a new weight $\bar{w}_\mathcal{P}$ where $\bar{w}_\mathcal{P} = \sum_k w_\mathcal{P}(k) \cdot |I_k|/(d_i-r_i)$. The routing path $\mathcal{P}_i$ for flow $j_i$ is then determined by randomly choosing a path $\mathcal{P}$ from $\mathcal{Q}_i$ using weight $\bar{w}_\mathcal{P}$ as the probability at which path $\mathcal{P}$ will be chosen. This path choosing process will be repeated for every flow. Consequently, a single path $\mathcal{P}_i$ will be determined for each flow $j_i \in \mathcal{J}$ and the packets from this flow will be routed through only this path.

Finally, we choose a transmission rate for each flow in every interval $I_k$. Denoting also $\mathcal{J}_e(k)$ the flow that will be transmitted on link $e$ in interval $I_k$, the transmission rate for every flow $j_i \in \mathcal{J}_e(k)$ will be set to $\sum_{j_i \in \mathcal{J}_e(k)} D_i$ and data packets from each flow in $\mathcal{J}_e(k)$ will be forwarded on $e$ using the EDF policy which we have introduced before.

\begin{algorithm}[!t]
\caption{\label{alg:approx} \textbf{Random-Schedule}}
\textbf{Input:} data center network $\mathcal{G}=(\mathcal{V},\mathcal{E})$, set of flows $\mathcal{J}$\\
\textbf{Output:} Routing path $\mathcal{P}_i$ for flows $j_i \in \mathcal{J}$ and transmission rate $s_i(t)$ for $t \in [r_i,d_i]$

\begin{algorithmic}[1]
\STATE Transform the DCFSR problem into a multi-step fractional multi-commodity flow problem by relaxing the constraints
\FOR{$I_k \in \{I_1,\ldots,I_K\}$}
	\STATE Solve the fractional multi-commodity flow problems by convex programming, obtaining $y^*_{i,e}(k)$
	\STATE Extract candidate paths for each flow, denote as $\mathcal{Q}_i(k)$ and a weight $w_\mathcal{P}(k)$ for each $\mathcal{P} \in \mathcal{Q}_i(k)$
\ENDFOR
\STATE $\mathcal{Q}_i = \cup_{1\leq k\leq K} \mathcal{Q}_i(k)$ for $j_i \in \mathcal{J}$
\STATE $\bar{w}_\mathcal{P} = \sum_k w_\mathcal{P}(k) \cdot |I_k|/(d_i-r_i)$ for  $\mathcal{P} \in \mathcal{Q}_i$
\FOR{$j_i \in \mathcal{J}$}
\STATE Randomly choose a path $\mathcal{P}_i$ from $\mathcal{Q}_i$ using weight $\bar{w}_\mathcal{P}$ as the probability
\ENDFOR
\STATE Route the packets from all the flows on link $e$ in interval $I_k$ using the EDF policy. The transmission rate for flow $j_i$ on link $e$ in interval $I_k$ is $\sum_{j_i \in \mathcal{J}_e(k)} D_i$.
\end{algorithmic}
\end{algorithm}

The whole process of the algorithm is shown in Algorithm~\ref{alg:approx}. We have to mention that the proposed algorithm does not guarantee the maximum operation rate constraint. However, we observe that the probability that many flows are simultaneously requested to be forwarded on a designated link in the proposed algorithm is very low as the probability for choosing a link for a flow is derived from the fractional solution which has the maximum operation rate constraint considered. Nevertheless, we can always repeat the randomized rounding process until we obtain a feasible solution. Now we show that 
\begin{theorem}
The deadline of every flow $j_i \in \mathcal{J}$ can be met in the solution produced by algorithm \textbf{Random-Schedule}. 
\end{theorem}
\begin{proof}
As we allow preemption for each flow, it suffices to show that all the data that arrives on a link $e$ in every interval $I_k \subseteq S_j$ can be transmitted by the end of this interval. Let us focus on one arbitrary link $e$ and an arbitrary interval $I_k$. The total amount of data that has to be transmitted through $e$ in this interval is given by $\sum_{j_i \in \mathcal{J}_e(k)} \left(D_i \cdot |I_k|\right)$. As the transmission rate for every flow in $\mathcal{J}_e(k)$ is $\sum_{j_i \in \mathcal{J}_e(k)} D_i$, the total time that is needed for accomplishing all the flows is equal to $\sum_{j_i \in \mathcal{J}_e(k)} \left(D_i \cdot |I_k|\right)/\sum_{j_i \in \mathcal{J}_e(k)} D_i = |I_k|$. As a result, all the data in this interval can always be transmitted no matter what kind of scheduling policy is used. However, we use the EDF policy because it can significantly reduce the frequency of changing the transmission rates of links.
\end{proof}

\subsection{Performance Analysis}

We now analyze the approximation performance of the proposed algorithm. Our results are based on the main results in \cite{Andrews_Fernandez-SS-2010} where the authors also used a rounding process to approximate the multi-commodity flow problem with $f(\cdot)$. The biggest difference compared with that work is that the rounding process we propose in this paper is responsible for minimizing the number of used links and thus has to guarantee the same path for each flow in every interval. That is, we aim at solving a multi-step MCF problem. We base our proof on the following result and we only show the difference from it.

\begin{theorem}[\cite{Andrews_Fernandez-SS-2010}]
\label{thm:mcf}
For nonuniform demands, randomized rounding can be used to achieve a $O\left(K+\log^{\alpha-1}D\right)-approximation$ for MCF with cost function $f(\cdot)$, where $K$ is the total number of demands and $D$ is the maximum demand among all the demands.
\end{theorem}

\begin{theorem}
 Algorithm~\textbf{Random-Schedule} can achieve a $O\left(\lambda^{\alpha}(n^2 \log D)^{\alpha-1}\right)$-approximation for the DCFSR problem with power function $g(x) = \mu x^{\alpha}$, where $D = \max_{j_i \in \mathcal{J}} D_i$.
\end{theorem}

\begin{proof}
We follow the process of the proof for Theorem~\ref{thm:mcf}. First, we assume unit flows, i.e., $w_i/(d_i-r_i) = 1$ and we use power consumption function $h(x) = \max\{\mu x, \mu x^{\alpha}\}$. Note here that we have $E(\hat{x}_e(\l)) \leq \lambda \sum_{k} \beta_k x^*_{e}(k)$ for any $1 \leq \l \leq K$. We then consider the following two cases:\\
Case $1$: $\sum_{1 \leq k\leq K} \beta_k x^*_{e}(k) \leq 1$, we have,
\begin{eqnarray}
&&E(\sum_{1\leq \l \leq K}g(\hat{x}_e(\l)) |I_{\l}|) \leq  \sum_{1\leq \l \leq K}\gamma_1 h(E(\hat{x}_e(\l)))|I_{\l}| \nonumber\\
							& \leq & \gamma_1 \lambda (t_K-t_0) \sum_{1 \leq k\leq K} \beta_k \mu x^*_{e}(k) \nonumber\\
							& = & \gamma_1 \lambda \sum_{1\leq k \leq K} h(x^*_e(k))|I_k| \\
							& = & \gamma_1 \lambda \sum_{1\leq k \leq K} \sum_{j_i \in J}h(x^*_{i,e}(k)) |I_k|, \nonumber
\end{eqnarray}
where the first inequality follows from the result in \cite{Andrews_Fernandez-SS-2010}.\\
Case $2$: $\sum_{1 \leq k\leq K} \beta_k x^*_{e}(k) > 1$, we have,
\begin{eqnarray}
&&E(\sum_{1\leq \l \leq K}g(\hat{x}_e(\l)) |I_{\l}|) \leq \sum_{1\leq \l \leq K}\gamma_2 h(E(\hat{x}_e(\l)))|I_{\l}| \nonumber\\
			&\leq& \gamma_2 \lambda^{\alpha} K^{\alpha-1} (t_K-t_0) \sum_{1 \leq k\leq K} \beta_k^{\alpha} \mu (x^*_{e}(k))^{\alpha} \nonumber\\
			&\leq& \gamma_2 \lambda^{\alpha} K \sum_{1\leq k \leq K} h(x^*_e(k))|I_k|  \\
			& \leq& 2\gamma_2 \lambda^{\alpha} n^{2(\alpha-1)} \sum_{1\leq k \leq K} \sum_{j_i \in J} h(x^*_{i,e}(k))|I_k|, \nonumber
\end{eqnarray}
where the first inequality follows also from the result in \cite{Andrews_Fernandez-SS-2010}. The second inequality and the last inequality follow from the property that $(x_1+\ldots+x_m)^{\alpha} \leq m^{\alpha-1}(x_1^{\alpha} + \ldots + x_m^{\alpha})$ while in the last inequality, we also apply $K \leq 2n$. The third inequality follows due to $\beta_k \leq 1$ for $1\leq k \leq K$. 

It can be observed that when the power consumption function is given by $g(x)=\mu x^{\alpha}$, $\sum_{j_i \in \mathcal{J}} h(x^*_{i,e}(k))|I_k|$ is a lower bound for the optimal energy consumption as it is the total energy consumed when the smallest transmission rate for each flow is used and also each flow can be routed through multiple paths. Consequently, we have
\begin{equation}
E(\sum_{1\leq \l \leq K}f(\hat{x}_e(\l)) |I_{\l}|)\leq O\left(\lambda^{\alpha} n^{2(n-1)}\right)\cdot \Phi_{opt}, 
\end{equation}
where the expression on the left side is the expectation of the energy consumption in the solution produced by Algorithm~\textbf{Random-Schedule}. Using Markov's inequality, the probability that the energy consumption is more than $c\cdot O(n^2)\cdot \Phi_{opt}$ is no more than $1/c$. This result then can be extended to nonuniform flows by introducing an extra factor $\log^{\alpha-1}D$, which has also been shown in \cite{Andrews_Fernandez-SS-2010}.
\end{proof}

\begin{theorem}
Algorithm~\textbf{Random-Schedule} can solve the DCFSR problem with the power function given in Eq.~\ref{eqn:power_func} while guaranteeing an approximation ratio $\gamma$ of $O\left(\lambda^{\alpha} (n^2\log D)^{\alpha-1}\right)$.
\end{theorem}
\begin{proof}
We also assume power consumption function $f(x)$ for the DCFSR problem and solve it with the proposed algorithm \textbf{Random-Schedule}. In the obtained fractional solution for the multi-step F-MCF problem, we use $z^*_{e}(k) \in \{0,1\}$ to indicate whether a link $e$ is chosen or not in interval $I_k$. We denote by $\hat{z}_e$ an indicator in the produced solution where $\hat{z}_e = 1$ if $e$ is used in at least one of the intervals; $\hat{z}_e = 0$ otherwise. Then, we have
\begin{eqnarray}
\sigma(t_K-t_0)E(\hat{z}_e) &=& \sigma(t_K-t_0) \left(1-\prod_{k}(1-E(\hat{z}_e(k)))\right) \nonumber\\
			&\leq & \sigma(t_K-t_0) \sum_{k} E(\hat{z}_e(k)) \nonumber \\
			&= & \sigma(t_K-t_0) \sum_{k} \left(1-\prod_{j_i \in J} (1-y^*_{i,e}(k))\right) \nonumber\\
			&\leq& \sigma(t_K-t_0) \sum_{k} \sum_{j_i \in J} y^*_{i,e}(k) \nonumber \\
			&\leq & \sigma(t_K-t_0) \sum_{k} nz^*_{e}(k)  \\
			&\leq & n\lambda\sum_{k} z^*_{e}(k)\sigma|I_k|, \nonumber
\end{eqnarray}
where the third inequality follows from $y^*_{i,e}(k) \leq z^*_e(k)$ and the last inequality follows from $\min_{k}|I_k| \leq |I_k|$ for any $1\leq k\leq K$. It can be observed that $\sum_{k} z^*_{e}(k)\sigma|I_k|$ is a lower bound for the optimal idle energy consumption by link $e$ since $z^*_{e}(k)$ is derived from the M-MCF problem which allows the links to be turned on and off freely at any moment. Combining all these together, we conclude that algorithm \textbf{Random-Schedule} can produce a $O\left(\lambda^{\alpha}(n^2\log D)^{\alpha-1}\right)$-approximation for the DCFSR problem.
\end{proof}

\subsection{Numerical Results}

We now briefly describe simulation results that illustrate the approximation performance of the proposed algorithm. We build a simulator with the \textbf{Random-Schedule} implemented in Python. We use the power consumption functions $x^{2}$ or $x^{4}$ and we choose a data center network topology which consists of $80$ switches (with $128$ servers connected). We consider [1, 100] as the time period of interest and as we assume no prior knowledge on the flows, we select release times and deadlines of flows randomly following a uniform distribution in [1,100].
The number of flows ranges from $40$ to $200$ and the amount of data from each flow is given by a random rational number following normal distribution $\mathcal{N}(10, 3)$. We compare three values of interest: lower bound (LB) for the optimum (solution given by $y^*_{i,v}(k)$), \textbf{Shortest-Path} (SP) routing plus \textbf{Most-Critical-First} (MCF), and \textbf{Random-Schedule} (RS). As SP is usually adopted, SP+MCF can give the lower bound of the energy consumption by SP routing, which represents the normal energy consumption in data centers. All of the values are normalized by the lower bound for the optimum and are averaged among $10$ independent runs. The simulation results are illustrated in Fig.~\ref{fig:sim}. As expected, RS outperforms SP+MCF to a large extent. Moreover, we notice that with the increase of the number of flows, the approximation ratio of RS converges while the approximation ratio of SP+MCF keeps an increasing trend. This confirms that combining routing and scheduling for flows can provide substantial improvements on the energy efficiency in data center networks.

This simulation serves merely as a primitive validation of the performance of the algorithm. Due to the space limit, we leave more exhaustive evaluation and further implementation as future work.

\begin{figure}[!t]
	\centering
	\subfigure{
 	\label{fig:sim-1} 
	\includegraphics[scale=0.67]{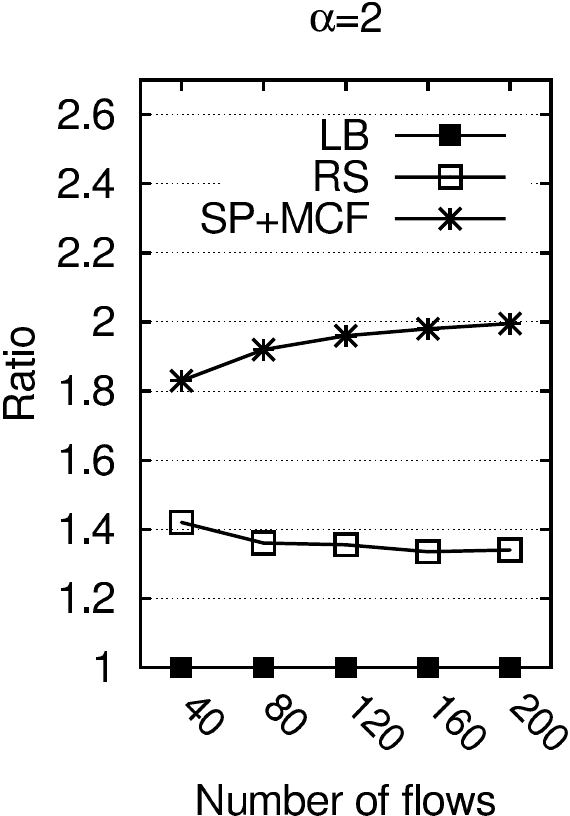}}
	\hspace{0in}
	\subfigure{
	\label{fig:sim-2} 
	\includegraphics[scale=0.67]{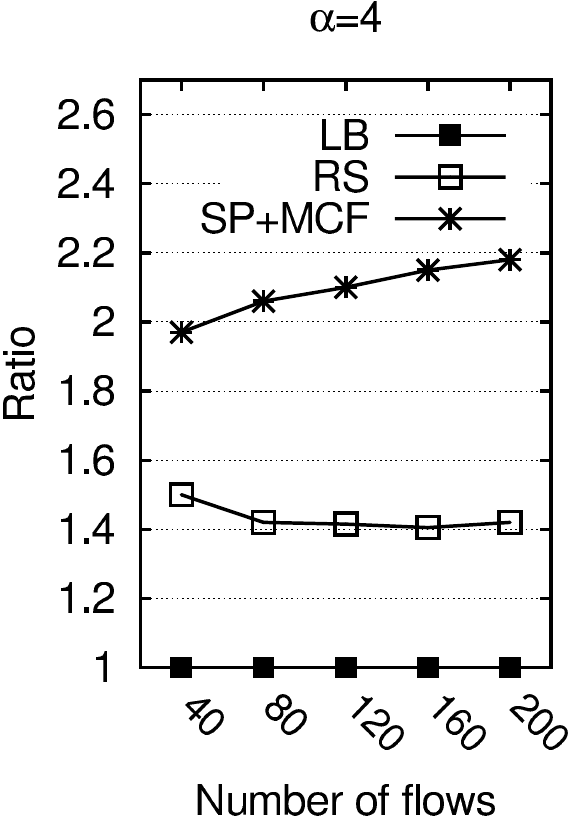}}
	\caption{The approximation performance of \textbf{Random-Schedule}.}
	\label{fig:sim} 
	\vspace{-0.45cm}
\end{figure}

\section{Related Work}
\label{sec:related}

This section summarizes some related work on the problem of improving the energy efficiency of DCNs, as well as the network scheduling and routing problem.

\subsection{Energy-Efficient Data Center Networks}

There has been a large body of work on improving the energy efficiency in DCNs. In general, they can be classified into two categories: The first line of work is designing new topologies that use fewer network devices while aiming to guarantee similar connectivity thus performance, such as the flatted butterfly proposed by Abts \emph{et al.} \cite{Abts_Marty-2010} or PCube \cite{Huang_Jia-2011}, a server-centric network topology for data centers, which can vary the bandwidth availability according to traffic demands. The second line of work is optimizing the energy efficiency by traffic engineering, i.e., consolidating flows and switching off unnecessary network devices. The most representative work in this category is ElasticTree \cite{Heller_Seetharaman}, which is a network-wide power manager that can dynamically adjust a set of active network elements to satisfy variable data center traffic loads. Shang \emph{et al.} \cite{Shang_Li} considered saving energy from a routing perspective, routing flows with as few network devices as possible. Mahadevan \emph{et al.} \cite{Mahadevan_banerjee-2011} discussed how to reduce the network operational power in large-scale systems and data centers. The first rate-adaptation based solution to achieve energy efficiency for future data center networks was provided by \cite{Wang_Zhang-2013}. Vasic \emph{et al.} \cite{Vasic_Bhurat-2011} developed a new energy saving scheme that is based on identifying and using energy-critical paths. Recently, Wang \emph{et al.} \cite{Wang_Yao} proposed CARPO, a correlation-aware power optimization algorithm that dynamically consolidates traffic flows onto a small set of links and switches and shuts down unused network devices. Zhang \emph{et al.} \cite{Zhang_Ansari-2012} proposed a hierarchical model to optimize the power in DCNs and proposed some simple heuristics for the model. In \cite{Wang_Zhang-2013-ICCCN}, the authors explored the problem of improving the network energy efficiency in MapReduce systems and afterwards in \cite{Wang_Zhang-PER-2013} they proposed to improve the network energy efficiency by joint optimizing virtual machine assignment and traffic engineering in data centers. Then, this method was extended to a general framework for achieving network energy efficiency in data centers  \cite{Wang_Zhang-JSAC-2013}. In \cite{Shang_li-2013}, the authors proposed to combine preemptive flow scheduling and energy-aware routing to achieve better energy efficiency. But performance guarantee was not considered. In a following work \cite{Xu_Shang-2013}, they considered greening data center networks by using a throughput-guaranteed power-aware routing scheme. To the best of our knowledge, the present paper is among the first to address the energy efficiency of DCNs from the aspect of scheduling and routing while guaranteeing the most critical performance criterion in DCNs - meeting flow deadlines. 

\subsection{Network Scheduling and Routing}

There have been a few works that have investigated the problem of job scheduling with deadlines in a multi-hop network. With known injection rate, a given deadline and fixed routes, the problem of online scheduling of sessions has been investigated by \cite{Andrews_Zhang-1999}, where they first gave a necessary condition for feasibility of sessions and gave an algorithm under which with high probability, most sessions are scheduled without violating the deadline when the necessary condition for feasibility is satisfied. In \cite{Mao_Koksal-2013}, the authors studied online packet scheduling with hard deadlines in general multihop communication networks. The algorithm they proposed gives the first provable competitive ratio, subject to hard deadline constraints for general network topologies. However, they didn't consider optimizing for energy efficiency.

Packet scheduling and routing for energy efficiency in general networks has also been well studied. In \cite{Andrews_Fernandez-SS-2010} and \cite{Andrews_Anta-pd}, the authors investigated to optimize the network energy efficiency from the aspect of routing and scheduling under continuous flow (transmission speed for each flow is given by a constant), by exploiting speed scaling and power-down strategy, respectively. Andrews \emph{et al.} \cite{Andrews_Antonakopoulos-2011} then proposed efficient packet scheduling algorithms for achieving energy efficiency while guaranteeing network stability. In \cite{Andrews_Zhang-2012}, they also provided efficient packet scheduling algorithms for optimizing the tradeoffs between delay, queue size and energy. 

Our approach has some fundamental differences with all the aforementioned solutions. Firstly, we combine speed scaling and power-down strategies for network devices in a unified model. Secondly, we carry out optimization from the granularity of flow instead of the granularity of packet and we aim at guaranteeing flow deadlines. Lastly, we investigate the problem of achieving energy efficiency by combining both flow scheduling and routing where we have to decide not only the routing path and schedule, but also the transmission rate for each flow.

\section{Conclusions}
\label{sec:conc}

In this paper, we studied flow scheduling and routing problem in data center networks where the deadlines of flows were strictly constrained and the objective was to minimize the energy consumption for transmitting all of the flows. The key observation in this work was that energy efficiency cannot be separately considered regardless of network performance being meeting flow deadlines the most critical requirement for it. We focused on two general versions of this problem with only scheduling and both routing and scheduling respectively. We introduced an optimal combinatorial algorithm for the version with only flow scheduling and devised an efficient approximation algorithm for the version with both routing and scheduling for flows, obtaining a provable performance ratio. With the proposed algorithms, we were able to achieve the aforementioned objective.


\section*{Acknowledgment}

This research was partially supported by National Science Foundation of China (NSFC) Major International Collaboration Project 61020106002, NSFC \& Hong Kong RGC Joint Project 61161160566, NSFC Project for Innovation Groups 61221062, NSFC Project grant 61202059 and 61202210, the Comunidad de Madrid grant S2009TIC-1692, and Spanish MICINN/MINECO grant TEC2011-29688-C02-01.




%
%
%

\balance

\end{document}